\documentclass[journal]{IEEEtran}
\usepackage{amsfonts,yfonts,eufrak,bm,multirow,mathtools,dcolumn}
\usepackage{amsmath}
\usepackage{amssymb}
\usepackage{graphicx}%
\newtheorem{theorem}{Theorem}

\newcommand{\bG}{{\mathbf G}}
\newcommand{\bH}{{\mathbf H}}

\newcommand{\bv}{{\mathbf v}}

\newcommand{\GF}[1] {{\rm GF} ( #1)}
\newcommand{\EG}[2] {{\rm EG} ( #1,#2)}
\newcommand{\PG}[2] {{\rm PG} ( #1,#2)}

\def\bbC{\mathbb{C}}

\def\bbZ{\mathbb{Z}}

\def\b0{\mathbf{0}}

\def\cC{{\cal C}}

\def\cG{{\cal G}}

\def\cS{{\cal S}}        \def\cT{{\cal T}}

\begin{document}
\title{High performance entanglement-assisted quantum LDPC codes need little entanglement}
\author{Min-Hsiu Hsieh, Wen-Tai Yen, and Li-Yi Hsu\thanks{Min-Hsiu Hsieh is with the ERATO-SORST
Quantum Computation and Information Project, Japan Science and Technology
Agency 5-28-3, Hongo, Bunkyo-ku, Tokyo, Japan. Wen-Tai Yen, and Li-Yi Hsu
conducted this research with Department of Physics, Chung Yuan Christian
University, Chungli, Taiwan.}}

%

\maketitle

\begin{abstract}
Though the entanglement-assisted formalism provides a universal connection between a
classical linear code and an entanglement-assisted quantum error-correcting code (EAQECC), the issue of maintaining large amount of pure maximally entangled states in constructing EAQECCs is a practical obstacle to its use. It is also conjectured that the power of entanglement-assisted formalism to convert those good classical codes comes from massive consumption of maximally entangled states. We show that the above conjecture is wrong by providing families of EAQECCs with an entanglement consumption rate that diminishes linearly as a function of the code length. Notably, two families of EAQECCs constructed in the paper require only one copy of maximally entangled state no matter how large the code length is. These families of EAQECCs that are constructed from classical finite geometric LDPC codes perform very well according to our numerical simulations. Our work indicates that EAQECCs are not only theoretically interesting, but also physically implementable. Finally, these high performance entanglement-assisted LDPC codes with low entanglement consumption rates allow one to construct high-performance standard QECCs with very similar parameters.

\end{abstract}


\section{Introduction}
The goal of coding theory is to design families of codes with transmission rate
approaching the channel capacity \cite{Shannon48}, while the error probability
of the transmitted message is arbitrarily small. Practical encoding and
decoding implementation is also desirable. Originally, Shannon employed
nonconstructive random codes with no practical encoding and decoding algorithm.
It is not surprising that most of the families of the constructed codes so far
do not satisfy both of the requirements. Exceptions are the low-density
parity-check (LDPC) codes \cite{RG63thesis} and Turbo codes \cite{BGT93}.

The LDPC code was first proposed in 1963 \cite{RG63thesis} that was much
earlier than the formation of modern coding theory. Not until the early 90's,
the LDPC code was rediscovered as family of sparse codes \cite{mackay99good},
and was shown to have capacity-approaching performance while the complexity of
implementing encoding and decoding algorithms is relatively low. A
$(J,L)$-regular LDPC code is defined to be the null space of a binary parity
check matrix $H$ with the following properties: (1) each column consists of $J$
``ones''; (2) each row consists of $L$ ``ones''; (3) both $J$ and $L$ are small
compared to the length of the code $n$ and the number of rows in $H$.

There are several methods of constructing good families of regular LDPC codes
\cite{mackay99good,KLF01,Fossorier04}. Among them, the LDPC codes that are
constructed from finite geometry have the following advantages: (1) they have
good minimum distance; (2) the girth of these codes is at least 6; (3) they
perform very well with iterative decoding, only a few tenths of a dB away from
the Shannon theoretical limit; (4) they can be put into either cyclic or
quasi-cyclic form. Consequently, their encoding can be achieved in linear time
and implemented with a single feedback shift register; (5) they can be extended
or shortened in various ways to obtain other good LDPC codes \cite{KLF01}.


The connection between classical linear codes and the quantum codes is unified
by the entanglement-assisted coding theory \cite{BDH06IEEE,BDH06,Hsieh08}.
Every classical linear code can be used to construct the corresponding quantum
code with the help of a certain amount of pre-shared entanglement. When the
classical code is self-dual, the resulting EAQECC is equivalent to a stabilizer
code \cite{DG97thesis}. Furthermore, the entanglement-assisted formalism
preserves the minimum distance property of the classical code---large minimum
distance classical code results in an entanglement-assisted quantum
error-correcting code (EAQECC) with the same minimum distance. However, it
is conjectured that the power of importing those good classical codes comes
from massive consumption of ebtis (maximally entangled states). The EAQECCs
constructed in Ref. \cite{HBD08QLDPC} seem to support the above conjecture since the ebits required grow with the code length n, and induce new criticism that EAQECCs
are of no practical use since maintaining so many noiseless ebits in EAQECCs
is extremely difficult.

In this paper, we show that the above conjecture is wrong by constructing families of EAQECCs from two types of finite geometries: the Euclidean geometry and the projective geometry. Moreover, we show that the pre-shared entanglement that are required  decreases linearly with respect to the length of the code. Notably, two families of EAQECCs constructed in the paper require only one copy of maximally entangled state no matter how large the code length is. We evaluate their
block error probability performance over the depolarizing channel when decoding
with the sum-product algorithm (SPA). These families of EAQECCs perform very well according to our numerical simulations.
Our work indicates that EAQECCs are no longer infeasible, and become a strong candidate for practical applications.

This paper is organized as follows.  In Section~\ref{II}, we first introduce
the Euclidean geometry and the projective geometry. Then, we discuss several
properties of these two finite geometries and show how to construct classical
finite geometric LDPC (FG-LDPC) codes. In Section~\ref{III}, we first review the definitions of standard QECCs and entanglement-assisted QECCs. We then construct
EAQECCs from classical FG-LDPC codes. Specifically, we construct several
EAQECCs that require an arbitrarily small amount of entanglement. We also show
that we can construct high performance standard QECCs with the catalytic
construction. In Section~\ref{IV}, we compare the performance of the EAQECCs
constructed from classical FG-LDPC codes with the known results in the
literature. In section~\ref{V}, we conclude.

\section{Finite geometry and finite geometry LDPC codes}\label{II}
In this section, we give definitions of finite geometries and show how to
construct classical FG-LDPC codes. Ref.~\cite{KLF01} contains an excellent
introduction of the Euclidean and projective geometries.

A finite geometry $\bG$ with $n$ points and $m$ lines is said to have the
following fundamental structural properties: (1) every line consists of $L$
points; (2) any two points are connected by one and only one line; (3) Every
point is intersected by $J$ lines; (4) two lines are either parallel or they
intersect at one and only one point. There are two families of finite
geometries which have the above properties, the Euclidean and projective
geometries over finite fields.


\subsection{Euclidean geometry}
Let $\EG{p}{q}$ be a $p$-dimensional Euclidean geometry over the Galois field
$\GF{q}$ where $p$ and $q$ are two positive integers. This geometry consists of
$q^p$ points, and each point can be represented by a $p$-tuple over $\GF{q}$.
The all-zero $p$-tuple $\b0=(0,0,\cdots,0)$ is called the origin. In other
words, all the points in ${\rm EG}(p,q)$ form a $p$-dimensional vector space
over $\GF{q}$. A line in $\EG{p}{q}$ can be viewed as a one-dimensional
subspace of $\EG{p}{q}$ or a coset of it. Therefore, a line in $\EG{p}{q}$
consists of $q$ points. Furthermore, the Euclidean geometry has the following
properties: (1) there are $q^{p-1}(q^{p}-1)/(q-1)$ lines; (2) for any point in
$\EG{p}{q}$, there are $(q^p-1)/(q-1)$ lines intersecting it; (3) every line
has $q^{p-1}-1$ lines parallel to it.

\subsubsection{Type-I EG-LDPC}
To show how to construct a binary parity check matrix using the Euclidean
geometry, we need a few definitions. Let $\GF{q^p}$ be the {\it extension
field} of $\GF{q}$. Then every point in $\EG{p}{q}$ is an element of the Galois
field $\GF{q^p}$, henceforth $\GF{q^p}$ can be regarded as the Euclidean
geometry $\EG{p}{q}$. Let $\alpha$ be a primitive element of $\GF{q^p}$. Then
$0,1,\alpha,\alpha^2,\cdots,\alpha^{q^p-2}$ can be mapped to each of the $q^p$
points in $\EG{p}{q}$.

Let $\bH_{\EG{p}{q}}^{(1)}$ be the binary matrix whose rows are the incidence
vectors of all the lines in $\EG{p}{q}$ that do not pass through the origin and
whose columns are the $q^p-1$ non-origin points. The columns are arranged in
the order of $1,\alpha,\alpha^2,\cdots,\alpha^{q^p-2}$, i.e., the $(i+1)$-th
column corresponds to the point $\alpha^i$. Then $\bH_{\EG{p}{q}}^{(1)}$ has
$n=q^p-1$ columns and $m=(q^{p-1}-1)(q^p-1)/(q-1)$ rows. To sum up, the binary
matrix $\bH_{\EG{p}{q}}^{(1)}$ has the following structural properties: (1)
each row has weight $L=q$. This correspondence results from each line of
$\EG{p}{q}$ containing $q$ points; (2) each column has weight
$J=(q^p-1)/(q-1)-1$. This correspondence results from the fact that each point
has $(q^p-1)/(q-1)$ lines intersecting at this point, but one of them passes
through the origin; (3) any two columns have at most one nonzero element in
common; (4) any two rows have at most one nonzero element in common; (5) the
density of $\bH_{\EG{p}{q}}^{(1)}$ is
$$\frac{L}{n}=\frac{J}{m}=\frac{q}{q^p-1}.$$ We can make the density smaller
by picking larger $p$ and $q$; (6) when $q$ is even, the minimum distance of the code defined by
$\bH_{\EG{p}{q}}^{(1)}$ is at least $J+1$. This can be proved using the BCH-bound
\cite{PW72}.

To be more specific, suppose $\ell$ is a line not passing through $\b0$. We can
define the incidence vector of $\ell$ (the $\ell$-th row in
$\bH_{\EG{p}{q}}^{(1)}$) as
$$\bv_\ell=(v_1,v_2,\cdots,v_n),$$
where $v_i=1$ if the point $\alpha^i$ lies in the line $\ell$,
otherwise $v_i=0$. Clearly, $\alpha^k \ell$ is also a line in
$\EG{p}{q}$, for $k=0,1,n-1$, and $\alpha\bv_\ell$ is a right
cyclic-shift of $\bv_\ell$.

Consider the respective incidence vectors of lines $\ell_j$, $\alpha\ell_j$,
$\cdots$, $\alpha^{n-1}\ell_j$. We can construct a binary $n\times n$ matrix
$H_{j}$ from them as follows:
\begin{equation}
H_j\equiv\left(
\begin{array}{c}
\bv_{\ell_j} \\
\bv_{\alpha\ell_j} \\
\vdots \\
\bv_{\alpha^{n-1}\ell_j}
\end{array}%
\right).
\end{equation}%
Here $H_j$ is a circulant matrix with column and row weights equal to $q$.
Since the total number of lines in $\EG{p}{q}$ not passing through $\b0$ is
$(q^p-1)(q^{p-1}-1)/(q-1)$, we can partition these lines into
$(q^{p-1}-1)/(q-1)$ cyclic classes (each cyclic class is represented by a
binary $n\times n$ cyclic matrix $H_j$). Finally, we can construct
$\bH_{\EG{p}{q}}^{(1)}$ by
\begin{equation}
\bH_{\EG{p}{q}}^{(1)}=\left(
\begin{array}{c}
H_1 \\
H_2 \\
\vdots \\
H_{\frac{q^{p-1}-1}{q-1}}
\end{array}%
\right).
\end{equation}%
The null space of $\bH_{\EG{p}{q}}^{(1)}$ is a type-I
EG-LDPC code. Furthermore, type-I EG-LDPC codes are cyclic codes.

\subsubsection{Type-II EG-LDPC}
The type-II EG-LDPC is obtained by taking the transpose of
$\bH_{\EG{p}{q}}^{(1)}$:%
\begin{equation*}
\bH_{\EG{p}{q}}^{(2)}\equiv\left(\bH_{\EG{p}{q}}^{(1)}\right)^T
=\left(\begin{array}{cccc}
H_1^T & H_2^T & \cdots &
H_{\frac{q^{p-1}-1}{q-1}}^T\end{array}\right).
\end{equation*} %
The null space of $\bH_{\EG{p}{q}}^{(2)}$ is a type-II
EG-LDPC code. Clearly, type-II EG-LDPC codes are quasi-cyclic codes.

The binary matrix $\bH_{\EG{p}{q}}^{(2)}$ has the following structural
properties: (1) each column has weight $J=q$; (2) each row has weight
$L=(q^p-1)/(q-1)-1$; (3) any two columns have at most one nonzero element in
common; (4) any two rows have at most one nonzero element in common; (5) the
density of $\bH_{\EG{p}{q}}^{(2)}$ is
$$\frac{L}{n}=\frac{J}{m}=\frac{q}{q^p-1};$$
(6) when $q$ is even, the minimum distance of the code defined by $\bH_{\EG{p}{q}}^{(2)}$ is at least $J+1$.

\subsection{Projective geometry}
Let $\GF{q^{p+1}}$ be the {\it extension field} of $\GF{q}$, and let $\alpha$
be a primitive element of $\GF{q^{p+1}}$. Let $n=(q^{p+1}-1)/(q-1),$ and
$\beta=\alpha^n$. Then the order of $\beta$ is $q-1$, and
$\{0,1,\beta,\cdots,\beta^{q-2}\}$ form all elements of $\GF{q}$. Consider the
set $\{\alpha^0,\alpha^1,\cdots, \alpha^{n}\}$, and partition the nonzero
elements of $\GF{q^{p+1}}$ into $n$ disjoint subsets as follows:
$$(\alpha^j)=\{\alpha^j,\beta\alpha^j,\cdots,\beta^{q-2}\alpha^j\},$$
for $j=0,1,\cdots,n-1$. Therefore, for any
$\alpha^i\in\GF{q^{p+1}}$, if $\alpha^i=\beta^\ell\alpha^j$ with
$0\leq j<n$, then $\alpha^i$ is in the set $(\alpha^j)$.

If we represent each element in $\GF{q^{p+1}}$ as an $(p+1)$-tuple
over $\GF{q}$, then $(\alpha^j)$ consists of $q-1$ $(p+1)$-tuples
over $\GF{q}$.

Define $\PG{p}{q}$ to be a $p$-dimensional projective geometry over $\GF{q}$.
This geometry consists of $n=(q^{p+1}-1)/(q-1)$ points, and each point is
represented by $(\alpha^j)$, for $0\leq j<n$. In other words, these $q-1$
elements, $\{\alpha^j,\beta\alpha^j,\cdots,\beta^{q-2}\alpha^j\}$, of
$\GF{q^{p+1}}$ is considered as the same point in $\PG{p}{q}$. Therefore, these
points, $(\alpha^0),(\alpha^1),\cdots,(\alpha^{n})$, form a $p$-dimensional
projective geometry over $\GF{q}$. Note that a projective geometry does not
have a origin. The projective geometry has the following properties: (1) each
line in $\PG{p}{q}$ consists of $q+1$ points; (2) the number of lines in
$\PG{p}{q}$ that intersect at a given point is $(q^p-1)/(q-1)$; (3) there are
$$\frac{(1+q+\cdots+q^{p-1})(1+q+\cdots+q^p)}{q+1}$$ lines in $\PG{p}{q}$.

\subsubsection{Type-I PG-LDPC}
Let $\bH_{\PG{p}{q}}^{(1)}$ be the binary matrix whose rows are the incidence
vectors of all lines in $\PG{p}{q}$ and whose columns are all the points of
$\PG{p}{q}$. The columns are arranged in the following order:
$(\alpha^0),(\alpha),\cdots,(\alpha^{n-1})$. Then $\bH_{\PG{p}{q}}^{(1)}$ has
$n=(q^{p+1}-1)/(q-1)$ columns and
$m=(1+q+\cdots+q^{p-1})(1+q+\cdots+q^p)/(q+1)$ rows. To sum up, the binary
matrix $\bH_{\PG{p}{q}}^{(1)}$ has the following structural properties: (1)
each row has weight $L=q+1$. This correspondence results from each line in
$\PG{p}{q}$ containing  $q+1$ points; (2) each column has weight
$J=(q^p-1)/(q-1)$. This correspondence results from the fact that each point
has $(q^p-1)/(q-1)$ lines intersecting at these points; (3) any two columns have
at most one nonzero element in common; (4) any two rows have at most one
nonzero element in common; (5) the density of $\bH_{\PG{p}{q}}^{(1)}$ is
$$\frac{L}{n}=\frac{J}{m}=\frac{q^2-1}{q^{p+1}-1}.$$
We can make the density smaller by picking $p\geq2$; (6) when $p=2$ and $q=2^s$ for $s\geq2$, the minimum distance of the code defined by $\bH_{\PG{p}{q}}^{(1)}$ is $J+1$. This can be proved using BCH-bound \cite{PW72}. Similar to the type-I EG-LDPC, the type-I PG-LDPC code is also cyclic.

\subsubsection{Type-II PG-LDPC}
The type-II PG-LDPC is obtained by taking the transpose of
$\bH_{\PG{p}{q}}^{(1)}$:
\begin{equation}
\bH_{\PG{p}{q}}^{(2)}\equiv\left(\bH_{\PG{p}{q}}^{(1)}\right)^T
\end{equation} %
The null space of $\bH_{\PG{p}{q}}^{(2)}$ is called the type-II
PG-LDPC code. Clearly, type-II PG-LDPC codes are quasi-cyclic codes.

The binary matrix $\bH_{\PG{p}{q}}^{(2)}$ has the following structural
properties: (1) each column has weight $J=q+1$; (2) each row has weight
$L=\frac{q^p-1}{q-1}$; (3) any two columns have at most one nonzero element in
common; (4) any two rows have at most one nonzero element in common; (5) the
density of $\bH_{\PG{p}{q}}^{(2)}$ is
$$\frac{L}{n}=\frac{J}{m}=\frac{q^2-1}{q^{p+1}-1};$$
(6) when $q$ is even, the minimum distance of the code defined by $\bH_{\PG{p}{q}}^{(2)}$ is
$J+1$.

\section{Construction of EAQECCs}\label{III}
\subsection{Entanglement-assisted formalism}
Denote by $\Pi$ the set of Pauli matrices $I,X,Y,Z$, where
$$
I = \left[\begin{array}{cc} 1 & 0 \\ 0 & 1 \end{array}\right], \quad
X = \left[\begin{array}{cc} 0 & 1 \\ 1 & 0 \end{array}\right], \quad
$$
$$
Y = \left[\begin{array}{cc} 0 & -i \\ i & 0
\end{array}\right],  \quad
Z = \left[\begin{array}{cc} 1 & 0 \\ 0 & -1 \end{array}\right].
$$ Define an $n$-fold Pauli operator $A=A_1\otimes A_2\otimes\cdots\otimes A_n$, where $A_i\in\Pi$ for $i=1,\cdots,n$. Define $[A]=\{\alpha A:\alpha\in\bbC, |\alpha|=1\}$. Let $\cG_n=\{[A]\}$ be the collection of all $n$-fold Pauli operators up to the equivalent class $[A]$. Then $\cG_n$ forms a multiplicative group under the following multiplication $[A][B]=[AB]$, where $A$ and $B$ are $n$-fold Pauli operators.

An $[[n,k]]$ standard QECC $\cC$ is defined to be a subspace  of dimension $2^k$ in the Hilbert space $\bbC_2^{\otimes n}$, where $\bbC_2$ is a two-dimensional space over the complex field. If $\cC$ can be specified as the $+1$ eigenspace of a set of commuting operators in $\cG_n$, $\cC$ is also known as a stabilizer code \cite{DG97thesis}. The group $\cS$ generated by the set of commuting operators in $\cG_n$ is called the stabilizer group.

A symplectic vector $\bm{\alpha}=(\bm{x}|\bm{z})$ is a binary vector of length $2n$, where $\bm{x}=(x_1,\cdots,x_n)$ and $\bm{z}=(z_1,\cdots.z_n)$ are $n$-bit strings with $x_i,z_i\in\bbZ_2=\{0,1\}$. Each symplectic vector $\bm{\alpha}$ corresponds to an $n$-fold Pauli operator $N_{\bm{\alpha}}$ in $\cG_n$:
\[
N_{\bm{\alpha}}=[X^{\bm{x}}Z^{\bm{z}}],
\]
where $X^{\bm{x}}=X^{x_1}\otimes\cdots\otimes X^{x_n}$, and  $Z^{\bm{z}}=Z^{z_1}\otimes\cdots\otimes Z^{z_n}$.
Define the symplectic product $\odot$ of two symplectic vectors $\bm{\alpha}=(\bm{x}|\bm{z})$ and $\bm{\beta}=(\bm{x}'|\bm{z}')$ to be
\[
\bm{\alpha}\odot\bm{\beta}=\bm{x}\cdot\bm{z}'+\bm{z}\cdot\bm{x}'
\]
where $\cdot$ is the regular inner product of two binary vectors and $+$ is the addition operation in binary field. The symplectic product between two symplectic vectors $\bm{\alpha}$ and $\bm{\beta}$ characterizes the commutation relation between two $n$-fold Pauli operators $N_{\bm{\alpha}}$ and $N_{\bm{\beta}}$:
\[
N_{\bm{\alpha}}N_{\bm{\beta}}=(-1)^{\bm{\alpha}\odot\bm{\beta}}N_{\bm{\beta}}N_{\bm{\alpha}}.
\]
Therefore, the stabilizer group $\cS=\langle N_{\bm{\alpha}_1},\cdots,N_{\bm{\alpha}_{n-k}} \rangle$ used to define a standard QECC $\cC$ is equivalent to the subspace in $(\bbZ_2)^{2n}$ spanned by $\{\bm{\alpha}_1,\cdots,\bm{\alpha}_{n-k}\}$ where $\bm{\alpha}_i\odot\bm{\alpha}_j=0$, $\forall i,j=1,\cdots,n-k$.

On the other hand, if one can construct a set of symplectic vectors that are orthogonal with respect to the symplectic product, such a set of vectors can be used to construct a corresponding stabilizer group, and thus defines a standard QECC. The simplest way of doing this is to begin with an $[n,k]$ classical binary linear code whose parity check matrix $H$ satisfies $HH^T=0$, where $T$ denotes matrix transpose. Then it is easy to check that each row vector of the following matrix
\begin{equation}
\left(\begin{array}{c|c} H & \mathbf{0} \\ \mathbf{0} & H \end{array}\right)
\end{equation}
is orthogonal to each other with respect to the symplectic product. The corresponding stabilizer group constructed from these symplectic vectors can then be used to define an $[[n,2k-n]]$ QECC. The above construction was discovered separately by Calderbank and Shor \cite{CS96} and by Steane \cite{Ste96}, and is called the CSS construction ever since.

A striking feature of QECCs that does not appear in classical codes is that two quantum errors sometimes need not be distinguished in order to correct them. If two errors, $E_1,E_2\in\cG_n$, are related by an element $P$ in the stabilizer group, say $E_1=E_2P$, these two errors have the same effect on the code space $\cC$, and have the same error syndrome. The degenerate effect of a QECC thus calls for a completely different decoding strategy \cite{PC08QLDPC}.

The dual-containing property of the classical binary code in the CSS construction is essential and therefore posts a strict constrain to import those good classical codes. However, if we begin with an arbitrary non-commuting group $\cT\subset\cG_n$, we will show in the following that $\cT$ defines an EAQECC. Without loss of generality, the non-commuting group $\cT$ is generated by a set of canonical generators $\{N_{\bm{\alpha}_1},\cdots,N_{\bm{\alpha}_{s}},N_{\bm{\alpha}_{s+1}},N_{\bm{\beta}_{s+1}},
\cdots,N_{\bm{\alpha}_{s+e}},N_{\bm{\beta}_{s+e}}\}$, for some integers $s$ and $e$, where \cite{BDH06IEEE}
\begin{eqnarray}
\bm{\alpha}_i\odot\bm{\alpha}_j&=&0, \forall i,j \label{symp_basis1}\\
\bm{\beta}_i\odot\bm{\beta}_j&=&0, \forall i,j \label{symp_basis2}\\
\bm{\alpha}_i\odot\bm{\beta}_j&=&0, \forall i\neq j \label{symp_basis3}\\
\bm{\alpha}_i\odot\bm{\beta}_i&=&1, \forall i. \label{symp_basis4}
\end{eqnarray}
With the help of $e$ copies of the maximally entangled state, we can obtain a commuting group $\cT'\subset\cG_{n+e}$ from $\cT$ such that
\begin{multline*}
\cT'=\langle N_{\bm{\alpha}_1}\otimes I_e,\cdots,N_{\bm{\alpha}_{s}}\otimes I_e, N_{\bm{\alpha}_{s+1}}\otimes X_1,N_{\bm{\beta}_{s+1}}\otimes Z_1, \\
\cdots,N_{\bm{\alpha}_{s+e}}\otimes X_e,N_{\bm{\beta}_{s+e}}\otimes Z_e\rangle
\end{multline*}
where $I_e$ denotes $e$ tensor copies of the identity $I$ operator and $X_i$ is an $e$-fold Pauli operator whose $i$-th position is the Pauli $X$ matrix, and identity $I$ for the rest positions. Then $\cT'$ (therefore $\cT$) defines an $[[n,k;e]]$ EAQECC, where $k+s+e=n$ \cite{BDH06,BDH06IEEE,Hsieh08}. The above entanglement-assisted formalism allows us to begin with an arbitrary binary check matrix $H$ in the CSS construction. The following theorem can be used to decide the amount of maximally entangled states in constructing a EAQECC \cite{HBD07,WB08EAQECC,Wilde09EAQECC}.
\begin{theorem}
\label{ebit} Let $H$ be any binary parity check matrix with
dimension $(n-k)\times n$. We can obtain the corresponding
$[[n,2k-n+e;e]]$ EAQECC, where $e = {\rm rank}(H H^T)$ is the number
of ebits needed.
\end{theorem}
Notice that Theorem~\ref{ebit} only provides a general guideline for evaluating the amount of entanglement required. 

\subsection{Entanglement-assisted quantum finite geometry LDPC codes}
Recall that the girth of classical FG-LDPC codes is at least $6$ due to the
geometric structure of finite geometry \cite{KLF01}. This makes the
construction of standard quantum LDPC codes from the classical FG-LDPC codes impossible because the classical FG-LDPC codes do not contain their dual unless necessary modification is made \cite{Aly07QLDPC,HI07QLDPC}. Moreover, even though modification is applied in Ref. \cite{Aly07QLDPC,HI07QLDPC}, it is very likely that the minimum distance of the QECCs will degrade.


The above obstacles can be overcome if we allow entanglement assistance. The EAQECC constructed from a classical FG-LDPC code naturally preserves both the minimum distance and the girth of its classical counterpart. It is conjectured that these benefits come from massive uses of ebits. Furthermore, noiseless entanglement is a valuable resource, and protecting it from the environment requires extra error-correcting power. Therefore, it is desirable to use as small amount of entanglement in EAQECCs as possible.
Previous evidence indicates that the amount of entanglement in EAQECCs might increase linearly with the
code length \cite{HBD08QLDPC}. Here, we illustrate three families of EAQECCs
constructed from classical FG-LDPC codes such that the entanglement consumption
rate decreases as a function of the code length.

Define the {\it entanglement consumption rate} of an $[[n,k;e]]$ EAQECC to be
$e/n$. The first example follows from classical type-I 2-dimensional EG-LDPC codes
over Euclidean geometry $\text{EG}(2,2^s)$. Such type-I 2-D EG-LDPC code is an
$[n,k,d]$ linear code where $n=2^{2s}-1$, $n-k=3^s-1$, and $d=2^s+1$.
Furthermore, the parity check matrix $\bH_{\text{EG}(2,2^s)}^{(1)}$ has both
row weight $L$ and column weight $J$ equal to $2^s$ \cite{KLF01}.
\begin{theorem}
The rank of $\bH_{\emph{EG}(2,2^s)}^{(1)}(\bH_{\emph{EG}(2,2^s)}^{(1)})^T$ is
equal to $2^s$. \label{THM_EG_RANK}
\end{theorem}
\begin{IEEEproof}
Denote $\overline{\text{EG}}(2,2^s)$ to be the Euclidean geometry
$\text{EG}(2,2^s)$ where both the origin and the lines passing through it are
excluded. Then $\overline{\text{EG}}(2,2^s)$ contains $2^{2s}-1$ points and
$2^{2s}-1$ lines. Recall the definition of a line in Euclidean geometry
$\text{EG}(2,2^s)$ from Section~\ref{II}. Any line in $\overline{\text{EG}}(2,
2^s)$ induces a partition of $\overline{\text{EG}}(2, 2^s)$ into $2^s+1$ sets,
where each set $S_{i}$, $i=1,2,\cdots,2^s+1$, contains lines parallel to each
other. It is also easy to verify that the size of each $S_i$ is $2^s-1$. We
consider the following three cases:
\begin{enumerate}
\item Recall that the number of points on a line is $2^s$. Therefore, the
    overlapping of the number of ``ones'' in the incidence vector with
    itself is even, and the inner product of an incidence vector with
    itself is zero.
\item Since two different lines in the same set are parallel to each other,
    the overlapping of the number of ``ones'' in these two incidence
    vectors is zero. The inner product of these two incidence vectors is
    zero.
\item Since two arbitrary different lines in two different sets intersect
    at only one point, the overlapping of the number of ``ones'' in these
    two incidence vectors is one. The inner product of these two incidence
    vectors is one.
\end{enumerate}
Since the rows of $\bH_{\text{EG}(2,2^s)}^{(1)}$ come from all the incidence
vectors of those lines in $\overline{\text{EG}}(2,2^s)$, we can arrange the
rows in the order of the lines in $S_i$, where $i$ starts from 1 to $2^s+1$.
Then the matrix $\bH_{\text{EG}(2,2^s)}^{(1)}(\bH_{\text{EG}(2,2^s)}^{(1)})^T$
consists of $(2^s+1)\times(2^s+1)$ submatrices:
\begin{equation*}
\left(\begin{array}{cccc}
\b0 & \mathbf{1} & \cdots & \mathbf{1} \\
\mathbf{1}& \b0 & & \mathbf{1} \\
\vdots& & \ddots& \vdots\\
\mathbf{1} & \cdots& & \b0
\end{array}\right),
\end{equation*}
where each $\b0$ or $\mathbf{1}$ represents an all-zeros or all-ones matrix of
size $(2^s-1)\times(2^s-1)$, respectively. The rank of
$\bH_{\text{EG}(2,2^s)}^{(1)}(\bH_{\text{EG}(2,2^s)}^{(1)})^T$ is then equal to
$2^s$.
\end{IEEEproof}

Table~\ref{EG_EAQECC_TABLE} lists a set of $[[n,2k-n+e,d;e]]$ EAQECCs
\cite{BDH06,BDH06IEEE,Hsieh08} constructed from the classical type-I 2-D
EG-LDPC code whose parity check matrix $\bH_{\text{EG}(2,2^s)}^{(1)}$ has row
weight $L$ and column weight $J$. The entanglement consumption rate in this
case is
\begin{equation}
\frac{e}{n}=\frac{2^s}{2^{2s}-1}\approx\frac{1}{\sqrt{n}},
\end{equation}
which decreases approximately equal to $1/\sqrt{n}$.
\begin{table}[h]
\caption{Each row represents an $[[n,2k-n+e,d;e]]$ EAQECC, respectively, that
is constructed from the classical type-I 2-D EG-LDPC code whose parity check
matrix $\bH_{\text{EG}(2,2^s)}^{(1)}$ has row weight $L$ and column weight
$J$.}
\centering
\begin{tabular}{|ccccccc|}
  \hline
  $s$ & $n$     &  $k$     & $d$   & $L$   & $J$   & $e$  \\ \hline
  2 & 15    & 7     & 5   & 4   & 4   & 4 \\
  3 & 63    & 37    & 9   & 8   & 8   & 8 \\
  4 & 255   & 175   & 17  & 16  & 16  & 16 \\
  5 & 1023  & 781   & 33  & 32  & 32  & 32 \\
  6 & 4095  & 3367  & 65  & 64  & 64  & 64 \\
  7 & 16383 & 14197 & 129 & 128 & 128 & 128 \\
  \hline
\end{tabular}
\label{EG_EAQECC_TABLE}
\end{table}

The second example follows from classical type-I 2-dimensional PG-LDPC codes
over projective geometry $\text{PG}(2,2^s)$. Such type-I 2-D PG-LDPC code is an
$[n,k,d]$ linear code where $n=2^{2s}+2^s+1$, $n-k=3^s-1$, and $d=2^s+2$.
Furthermore, the parity check matrix $\bH_{\text{EG}(2,2^s)}^{(1)}$ has both
row weight $L$ and column weight $J$ equal to $2^s+1$ \cite{KLF01}.
\begin{theorem}
The rank of $\bH_{\emph{PG}(2,2^s)}^{(1)}(\bH_{\emph{PG}(2,2^s)}^{(1)})^T$ is
equal to $1$, $\forall s\in\bbZ^+$.
\end{theorem}
\begin{IEEEproof}
Recall that  $\text{PG}(2,2^s)$ contains $2^{2s}+2^s+1$ points and
$2^{2s}+2^s+1$ lines, and every line intersects with another one at exactly one
point (no parallel lines in $\text{PG}(p,q)$). The overlapping of the number of
``ones'' in these two incidence vectors is one. Therefore, the inner product of
these two incidence vectors is one. Furthermore, the number of points on a line
is $2^s+1$, the overlapping of the number of ``ones'' in the incidence vector
with itself is odd. The inner product of the incidence vector with itself is
one.

Since the rows of $\bH_{\text{PG}(2,2^s)}^{(1)}$ come from all the incidence
vectors of those lines in $\text{PG}(2,2^s)$, the matrix
$\bH_{\text{PG}(2,2^s)}^{(1)}(\bH_{\text{PG}(2,2^s)}^{(1)})^T$ is an all-one
matrix. The rank of
$\bH_{\text{PG}(2,2^s)}^{(1)}(\bH_{\text{PG}(2,2^s)}^{(1)})^T$ is then equal to
$1$.
\end{IEEEproof}

Table~\ref{TABLE_PG_EAQECC} lists a set of $[[n,2k-n+e,d;e]]$ EAQECCs
constructed from the classical type-I 2-D PG-LDPC code whose parity check
matrix $\bH_{\text{PG}(2,2^s)}^{(1)}$ has row weight $L$ and column weight $J$.
The entanglement consumption rate in this case is
\begin{equation}
\frac{e}{n}=\frac{1}{2^{2s}+2^s+1}=\frac{1}{n},
\end{equation}
which decreases linearly with respect to $n$.
\begin{table}[h]
\caption{Each row represents an $[[n,2k-n+e,d;e]]$ EAQECC, respectively, that
is constructed from the classical type-I 2-D PG-LDPC code whose parity check
matrix $\bH_{\text{PG}(2,2^s)}^{(1)}$ has row weight $L$ and column weight
$J$.}
\centering
\begin{tabular}{|ccccccc|}
  \hline
  $s$ & $n$     &  $k$     & $d$   & $L$   & $J$   & $e$  \\ \hline
  2 & 21    & 11    & 6   & 5   & 5   & 1 \\
  3 & 73    & 45    & 10  & 9   & 9   & 1 \\
  4 & 273   & 191   & 18  & 17  & 17  & 1 \\
  5 & 1057  & 813   & 34  & 33  & 33  & 1 \\
  6 & 4161  & 3431  & 66  & 65  & 66  & 1 \\
  7 & 16513 & 14326 & 130 & 129 & 129 & 1 \\
  \hline
\end{tabular}
\label{TABLE_PG_EAQECC}
\end{table}

The third example follows from classical type-II 3-dimensional PG-LDPC codes
over projective geometry $\text{PG}(3,q)$.
\begin{theorem}
The rank of $\bH_{\emph{PG}(3,q)}^{(2)}(\bH_{\emph{PG}(3,q)}^{(2)})^T$ is equal
to $1$, for every integer $q\geq2$.
\end{theorem}
\begin{IEEEproof}
Here, we consider a 3-dimensional projective geometry over GF$(q)$. Recall that
each line in $\text{PG}(3,q)$ contains $L=q^2+q+1$ points, where $L$ is odd for
$q\geq 2$. Therefore, the inner product of the row vector with itself is one.
Furthermore, two different points are connected by exactly one line. Therefore,
the overlapping of the number of ``ones'' in arbitrary two rows is one. The
inner product of these two incidence vectors is one. Therefore the matrix
$\bH_{\text{PG}(3,q)}^{(2)}(\bH_{\text{PG}(3,q)}^{(2)})^T$ is an all-one
matrix. The rank of $\bH_{\text{PG}(3,q)}^{(2)}(\bH_{\text{PG}(3,q)}^{(2)})^T$
is then equal to $1$.
\end{IEEEproof}

Table~\ref{TABLE_PG_EAQECC2} lists a set of $[[n,2k-n,d;e]]$ EAQECCs
constructed from the classical type-II 3-D PG-LDPC code whose parity check
matrix $\bH_{\text{PG}(3,q)}^{(2)}$ has row weigh $L$ and column weigh $J$.
Again the construction uses the ``generalized CSS construction'' proposed in
Ref.~\cite{BDH06,BDH06IEEE}. The entanglement consumption rate in this case is
\begin{equation}
\frac{e}{n}=\frac{q+1}{(1+q+q^2)(1+q+q^2+q^3)}=\frac{1}{n},
\end{equation}
which decreases linearly with respect to $n$.
\begin{table}[htbp]
\caption{Each row represents an $[[n,2k-n+e,d;e]]$ EAQECC, respectively, that
is constructed from the classical type-II 3-D PG-LDPC code whose parity check
matrix $\bH_{\text{PG}(3,q)}^{(2)}$ has row weight $L$ and column weight $J$.}
\centering
\begin{tabular}{|ccccccc|}
  \hline
  $q$ & $n$   & $k$   & $d$ & $L$ & $J$  & $e$  \\ \hline
  2   & 35    & 24    & 4   & 7   & 3    & 1 \\
  3   & 130   & 91    & 5   & 13  & 4    & 1 \\
  4   & 357   & 296   & 6   & 21  & 5    & 1 \\
  5   & 806   & 651   & 7   & 31  & 6    & 1 \\
  6   & 2850  & 2451  & 8   & 43  & 7    & 1 \\
  7   & 4745  & 4344  & 9   & 57  & 8    & 1 \\
  \hline
\end{tabular}
\label{TABLE_PG_EAQECC2}\end{table}

The investigation of the entanglement-assisted FG-LDPC codes also allows one to
construct high performance standard QECCs. The basic idea is simple. Instead of
sharing entanglement with the receiver Bob, Alice prepares $e$ ebits locally.
With $e$ halves she encodes using some $[[n,k,d;e]]$ entanglement-assisted
FG-LDPC codes. The other $e$ halves she encodes using a good simple standard
$[[n',e,d]]$ QECC. Putting both blocks together, she has an $[[n+n',k,d]]$
standard QECC \cite{BDH06IEEE}. The entanglement-assisted FG-LDPC code is an
excellent candidate for such construction, since the amount of entanglement
needed either grows slowly with $n$ or not at all, and with $e \ll n$, this
extra block likely has little effect on the overall code performance.
Therefore, having high performance entanglement-assisted LDPC codes with low
entanglement consumption rates implies that one can construct high-performance
standard QECCs with very similar parameters. Moreover, these codes are likely
to work much better than self-dual LDPC codes because they do not have
4-cycles.

\section{Performance}\label{IV}
In this section, we provide simulation results (in terms of block error rate)
of the entanglement-assisted FG-LDPC codes over the depolarizing channel, which
creates $X$ errors, $Y$ errors, and $Z$ errors with equal probability $f_m$.
Moreover, we focus on those entanglement-assisted FG-LDPC codes with a low
entanglement consumption rate. The decoding algorithm used in the simulation is
the sum-product decoding algorithm. For simplicity, we omit the introduction of
this decoding algorithm, and point the interested reader to Refs.
\cite{mackay99good,PC08QLDPC}.

Fig.~\ref{fig1} shows that the block error probability performance of EG(2,32)
is better than EG(2,16), and the block error probability performance of
EG(2,16) is better than EG(2,8) when the cross over probability $f_m$ is small
($f_m < 0.015$). A similar result holds for quantum PG-LDPC codes. However, the
block error probability performance for shorter code length is better when the
cross error probability is large. The reason for this might be because in the
quantum setting, the transmitted quantum information cannot be retrieved even
when the whole block contains just one uncorrectable error. In this sense,
using quantum code with large block in the very noisy channel might not be helpful, unlike in the classical setting.
\begin{figure}[htbp]
\includegraphics[width=0.5\textwidth]{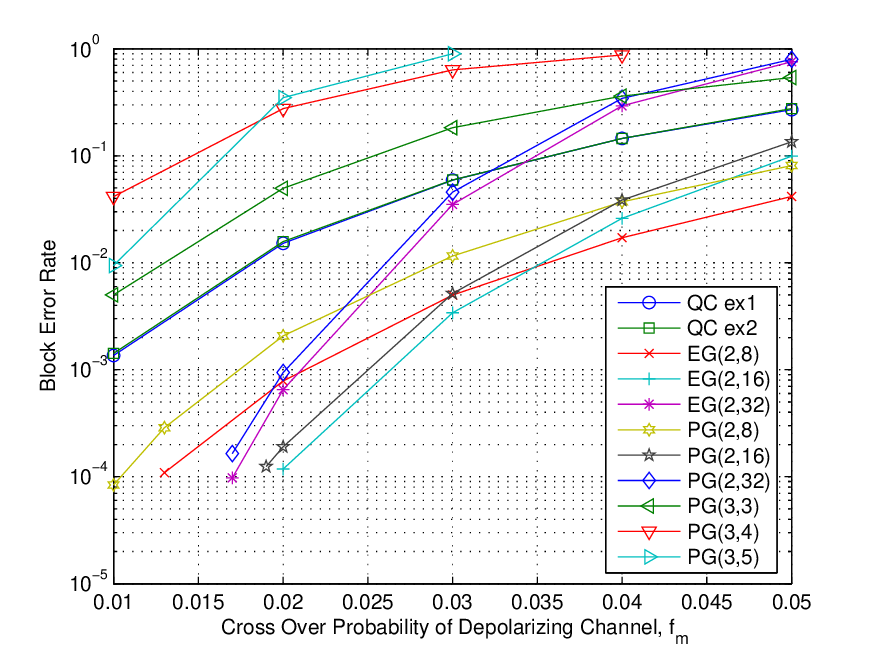}
\caption{(Color online). Block error probability performance of entanglement-assisted FG-LDPC codes with sum-product algorithm (SPA) decoding, and 100 iterations for each date point. ``QC ex1'' and ``QC ex2'' are the  entanglement-assisted quasi-cyclic LDPC codes constructed in Ref.~\cite{HBD08QLDPC}, respectively.}
\label{fig1}
\end{figure}

The authors in \cite{HBD08QLDPC} investigated the block error probability
performance of entanglement-assisted quantum quasi-cyclic LDPC codes, and
showed that their EAQECCs outperform some existing quantum stabilizer codes
with similar {\it net rate}, where the net rate of an $[[n,k;e]]$ EAQECC is
defined to be $(k-e)/n$. Surprisingly, the 2-D entanglement-assisted FG-LDPC
codes perform much better than their constructed examples. Moreover, the
consumed pure entanglement in constructing the entanglement-assisted FG-LDPC
codes is much less than theirs.

The authors in \cite{HI07QLDPC} proposed a construction of a pair of quasi-cyclic
LDPC codes to construct a quantum error correction code. We numerically simulate two such quantum quasi-cyclic LDPC codes with the following parameters defined in Theorem 2.4 in Ref.~\cite{HI07QLDPC}:
\begin{table}[h]\centering
\begin{tabular}{|c|c|c|}\hline
  & C=$(J,L,P,\sigma,\tau)$ & D=$(K,L,P,\sigma,\tau)$\\ \hline
 Code A $[[n=1168, k=590]]$& (4,16,73,10,2) & (4,16,73,10,2)\\ \hline
 Code B $[[n=1204,k=614]]$& (7,28,43,8,6) & (7,28,43,8,6)\\
  \hline
\end{tabular}
\end{table}

We compare their block error probability performance with EG(2,32) and PG(2,32) in Fig.~\ref{FIG_HI} since they have similar code length. Even though there is no 4-cycle in these two quantum quasi-cyclic LDPC codes, these quantum codes likely contain many low-weight codewords. On the other hand, the minimum distance of EG(2,32) and PG(2,32) is guaranteed from its geometric construction.
\begin{figure}[htbp]
\includegraphics[width=0.5\textwidth]{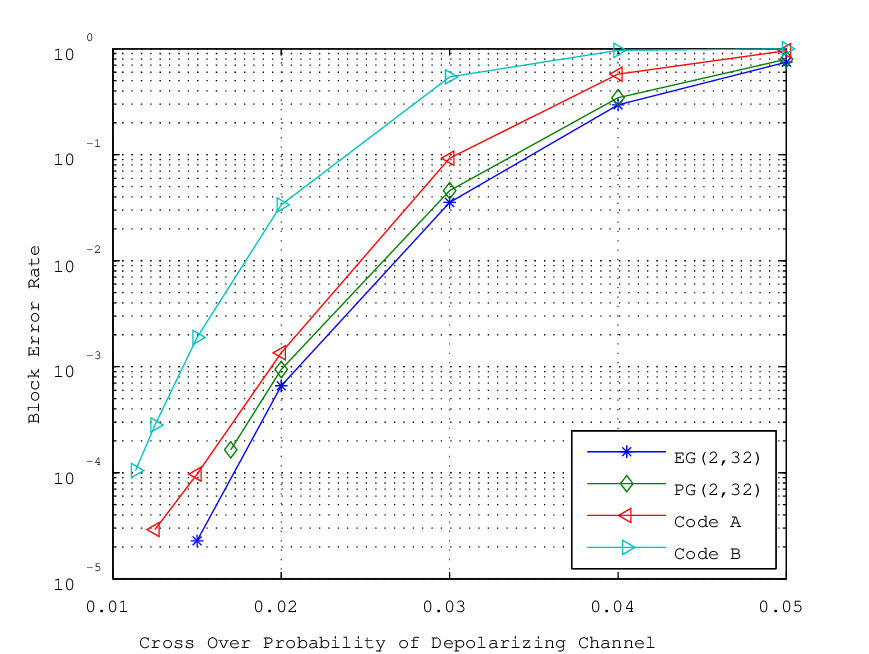}
\caption{(Color online). Block error probability performance of EG(2,32) and PG(2,32) constructed in this paper and Code A and Code B constructed using methods in
Ref.~\cite{HI07QLDPC}. SPA decoding with 100 iterations is used.}
\label{FIG_HI}
\end{figure}

The author in \cite{Aly07QLDPC} constructed a class of stabilizer quantum LDPC
codes from the Euclidean geometry. By adding a column of ``ones'' or a column
of ``ones'' together with an identity matrix to the classical EG-LDPC, the
resulting density matrix is self-dual. Similar technique can be applied to the
classical PG-LDPC codes. Therefore, the stabilizer quantum LDPC codes can be
constructed by the CSS construction method \cite{NC00}. As pointed out in
\cite{Aly07QLDPC}, these quantum LDPC codes contain only one cycle of length
four. We compare the block error probability performance of the stabilizer
quantum EG-LDPC codes constructed in \cite{Aly07QLDPC} with the
entanglement-assisted EG-LDPC codes proposed in this article in
Fig.~\ref{EG_ALY}. Even though there is only one cycle of length four in the conventional
quantum EG-LDPC codes, its performance is much worse than the
entanglement-assisted EG-LDPC codes. Fig.~\ref{PG_ALY} shows similar results
for the stabilizer quantum PG-LDPC codes and the entanglement-assisted PG-LDPC
codes. The degradation of the block error probability performance of the
stabilizer quantum FG-LDPC codes mainly comes from the single 4-cycle. Our
simulation shows that whenever errors occur on those qubits in that 4-cycle,
they are unlikely to be corrected by SPA decoding.

\begin{figure}[htbp]
\includegraphics[width=0.5\textwidth]{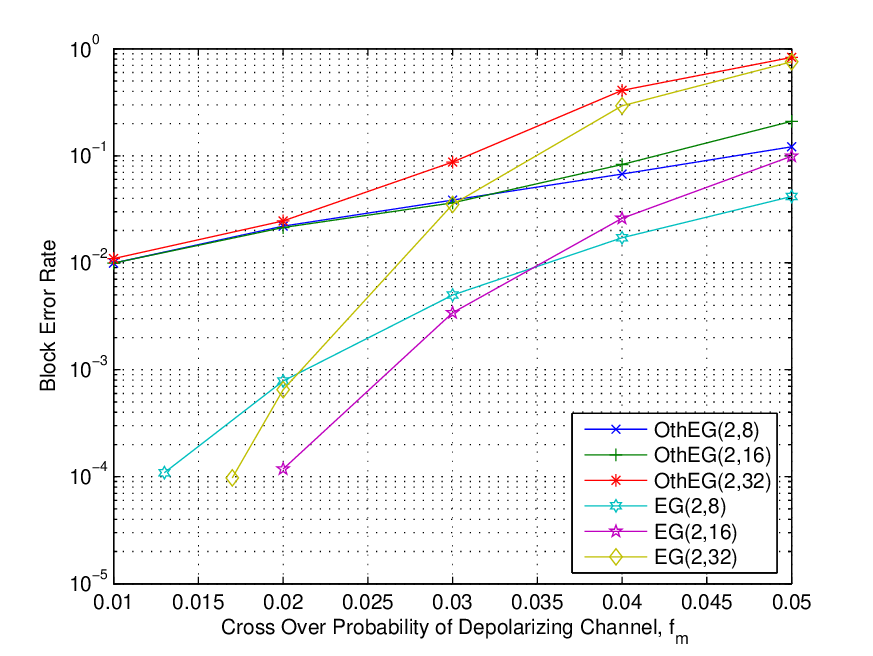}
\caption{(Color online).
We compare the block error probability performance of stabilizer quantum EG-LDPC codes (labeled by OthEG(p,q))
constructed in \cite{Aly07QLDPC} and the entanglement-assisted EG-LDPC codes (labeled by EG(p,q)) proposed in this article.
The SPA decoding algorithm is used in both simulations with 100 iterations for
each date point.} \label{EG_ALY}
\end{figure}

\begin{figure}[htbp]
\includegraphics[width=0.5\textwidth]{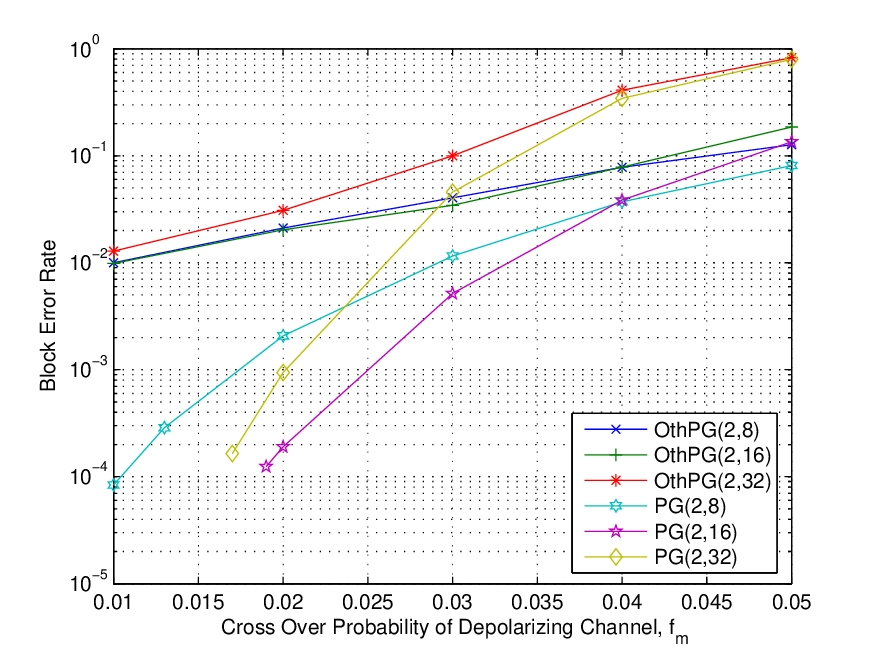}
\caption{(Color online).
We compare the block error probability performance of stabilizer quantum PG-LDPC codes (labeled by OthPG(p,q))
constructed in \cite{Aly07QLDPC} and the entanglement-assisted PG-LDPC codes (labeled by PG(p,q)) proposed in this article.
The SPA decoding algorithm is used in both simulations with 100 iterations for
each date point.}
\label{PG_ALY}
\end{figure}

Next, we modify the sum-product decoding algorithm according to the heuristic
methods proposed in \cite{PC08QLDPC}. Those modifications are intended to
overcome the ignorance of the degeneracy in the decoding. However, our
simulation shows that those modifications do not help to improve the
performance of decoding the entanglement-assisted FG-LDPC codes. For example,
in Fig. \ref{fig2}, we show the performance of the SPA decoding with random
perturbation (see Ref.~\cite{PC08QLDPC} for further detail) is the same as that
of no random perturbation for the EG(2,8) EAQECC. The reason for such result is
because the degeneracy effect is mild. Those low weight errors are not likely
to be inside the code space due to the large minimum distance property of the
FG-LDPC codes.

\begin{figure}[htbp]
    \includegraphics[width=0.5\textwidth]{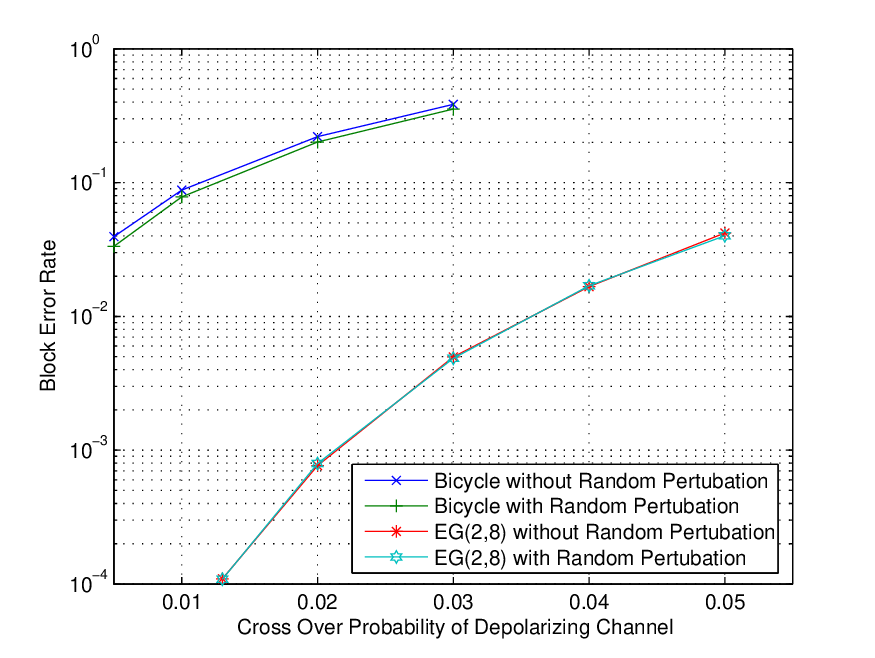}
\caption{(Color online). Performance of quantum FG-LDPC codes with modified SPA
decoding. The maximum number of iterations for the SPA decoding is 360, and the number of iterations between each perturbation is 60.
The strength of the random perturbation is 0.1. The stabilizer code constructed by the bicycle technique (see Ref.~\cite{MMM04QLDPC})
encodes 30 logical qubits in 60 physical qubits. }
  \label{fig2}
\end{figure}

\section{Conclusion}\label{V}
In this paper, we construct families of EAQECCs from finite geometries such
that the block error probability performance of these codes is relatively
better than those proposed in the literature so far. The improvement largely
comes from lack of cycles of length 4 of the constructed FG-LDPC codes due to
the geometric structure. Furthermore, we can overcome the problem of
maintaining large amount of pure maximally entangled states in constructing
EAQECCs by providing families of EAQECCs with an exponentially decreasing
entanglement consumption rate.

The degeneracy effect of the entanglement-assisted FG-LDPC codes is mild
because low weight errors are unlikely to be codewords due to the guaranteed
large minimum distance of the FG-LDPC codes. Therefore, we do not need to
modify the sum-product decoding algorithm which would largely increase the
decoding complexity. However, we believe that new decoding technique that
incorporates the coset construct of the quantum codes deserves further
investigation.

Having high performance entanglement-assisted LDPC codes with low entanglement
consumption rates implies that one can construct high-performance standard
QECCs with very similar parameters. This is because we can use a short and
simple stabilizer code to encode the other halves of entanglement, so long as
it has high enough distance and is easy to decode. This extra block stabilizer
code has little effect on the overall code performance. Moreover, these codes
are likely to work much better than self-dual LDPC codes because they do not
have 4-cycles.

\section*{Acknowledgments}
The author MHH thanks Todd Brun and Mark M. Wilde for their useful suggestions
and comments on the draft. The authors LYH and YWT acknowledge support from
National Science Council of the Republic of China under Contract No.
NSC.96-2112-M-033-007-MY3.
\bibliographystyle{unsrt}
\bibliography{../../Ref}

\end{document}